\documentclass[a4paper,UKenglish,cleveref, autoref, thm-restate]{lipics-v2021}
\nolinenumbers

\usepackage[algoruled,linesnumbered,algo2e,vlined]{algorithm2e}
\usepackage{amsmath}
\usepackage{enumerate}

\usepackage[numbers]{natbib}

\newcommand{\emptystr}{\varepsilon}

\newcommand{\nfreq}{\mathsf{NF}}
\newcommand{\nocc}{\mathsf{NOcc}}
\newcommand{\substr}{\mathcal{S}}
\newcommand{\rep}{\substr^{\mathsf{r}}}
\newcommand{\mr}{\substr^{\mathsf{mr}}}
\newcommand{\lmr}{\substr^{\mathsf{lmr}}}
\newcommand{\rmr}{\substr^{\mathsf{rmr}}}
\newcommand{\nsmr}{\substr^{\mathsf{nsmr}}}
\newcommand{\smr}{\substr^{\mathsf{smr}}}
\newcommand{\occ}{\mathsf{Occ}}
\newcommand{\pqcd}[2]{(#1, #2)}

\newcommand{\sufa}{\mathsf{SA}}

\newcommand{\bwt}{\mathsf{L}}
\newcommand{\LF}{\mathsf{LF}}
\newcommand{\FL}{\mathsf{FL}}
\newcommand{\lcontext}{\mathsf{lc}}
\newcommand{\rcontext}{\mathsf{rc}}
\newcommand{\runc}{\mathsf{L}'}
\newcommand{\rund}{\mathsf{d}}

\newcommand{\sufi}{\mathcal{I}}
\newcommand{\repr}{\mathsf{repr}}
\newcommand{\rext}{\mathsf{rlist}}
\newcommand{\revtrie}{\mathcal{T}}

\newcommand{\rdq}{\mathsf{RD}}

\newcommand{\idtt}[1]{\ensuremath{\mathtt{#1}}}

\usepackage{hyperref}
\usepackage{xcolor}

\title{R-enum Revisited: Speedup and Extension for Context-Sensitive Repeats and Net Frequencies}

\author{Kotaro Kimura}{Kyushu Institute of Technology, Japan}{kimura.kotaro365@mail.kyutech.jp}{}{}
\author{Tomohiro I}{Kyushu Institute of Technology, Japan}{tomohiro@ai.kyutech.ac.jp}{https://orcid.org/0000-0001-9106-6192}{KAKENHI 24K02899, JST AIP Acceleration Research JPMJCR24U4}

\titlerunning{R-enum Revisited: Speedup and Extension}

\authorrunning{K. Kimura and T. I}

\Copyright{Kotaro Kimura and Tomohiro I}



\begin{document}
\maketitle              
\begin{abstract}
  A \emph{repeat} is a substring that occurs at least twice in a string, and is called a \emph{maximal repeat} if it cannot be extended outwards without reducing its frequency.
  Nishimoto and Tabei [CPM, 2021] proposed \emph{r-enum}, an algorithm to enumerate various characteristic substrings, including maximal repeats, in a string $T$ of length $n$ in $O(r)$ words of compressed working space, where $r \le n$ is the number of runs in the Burrows-Wheeler transform (BWT) of $T$.
  Given the run-length encoded BWT (RLBWT) of $T$, r-enum runs in $O(n \log \log_{w} (n/r))$ time in addition to the time linear to the number of output strings, where $w = \Theta(\log n)$ is the word size.
  In this paper, we first improve the $O(n \log \log_{w} (n/r))$ term to $O(n)$.
  We next extend r-enum to compute other context-sensitive repeats such as \emph{near-supermaximal repeats (NSMRs)} and \emph{supermaximal repeats}, as well as the \emph{context diversity} for every maximal repeat in the same complexities.
  Furthermore, we study net occurrences:
  An occurrence of a repeat is called a \emph{net occurrence} if it is not covered by another repeat, and the \emph{net frequency} of a repeat is the number of its net occurrences.
  With this terminology, an NSMR is a repeat with a positive net frequency.
  Given the RLBWT of $T$, we show how to compute the set $\nsmr$ of all NSMRs in $T$ together with their net frequency/occurrences in $O(n)$ time and $O(r)$ space.
  We also show that an $O(r)$-space data structure can be built from the RLBWT to compute the net frequency/occurrences of any pattern in optimal time.
  The data structure is built in $O(r)$ space and in $O(n)$ time with high probability or deterministic $O(n + |\nsmr| \log \log \min(\sigma, |\nsmr|))$ time, where $\sigma \le r$ is the alphabet size of $T$.
  To achieve this, we prove that the total number of net occurrences is less than $2r$.
  With the duality between net occurrences and \emph{minimal unique substrings (MUSs)},
  we get a new upper bound $2r$ of the number of MUSs in $T$, which may be of independent interest.
\end{abstract}

\section{Introduction}\label{sec:intro}
Finding characteristic substring patterns in a given string $T$ is a fundamental task for string processing.
One of the simplest patterns based on frequencies would be a \emph{repeat}, which occurs at least twice in $T$.
Since the number of distinct repeats in $T$ of length $n$ can be $\Theta(n^2)$, maximality is often utilized to reduce the number down to $O(n)$~\cite{1997Gusfield_AlgorOnStrinTreesAnd}.
A substring of $T$ is called a \emph{right-maximal repeat (RMR)} (resp. \emph{left-maximal repeat (LMR)}) if it cannot be extended to the right (resp. to the left) without reducing its number of occurrences in $T$.
If a repeat is both right-maximal and left-maximal, it is called a \emph{maximal repeat (MR)}.

Gusfield~\cite{1997Gusfield_AlgorOnStrinTreesAnd} also introduced stricter criteria to reduce the redundancy in maximal repeats.
A repeat is called a \emph{supermaximal repeat (SMR)} if it is not a substring of another repeat,
and a \emph{near-supermaximal repeat (NSMR)} if at least one of its occurrences is not covered by another repeat.
It is easy to see that an SMR is an NSMR, and an NSMR is an MR.
NSMRs were studied under the name of \emph{Chinese frequent strings} in~\cite{2001LinY_ExtracChinesFrequenStrinWithout,2004LinY_ProperAndFurthApplicOf,2024OhlebuschBO_FasterComputOfChinesFrequen_SPIRE}, and \emph{largest maximal repeats} in~\cite{2008NicolasRSPCDTVM_ModelLocalRepeatOnGenom,2013Galle_BagOfRepeatRepresOf_SIGIR}.
Gall{\'{e}} and Tealdi~\cite{2014GalleT_ContexDiverRepeatAndTheir_LATA,2018GalleT_EmphxRepeatNewTaxonOf} proposed the concept of \emph{context diversities} to define a general class of context-sensitive repeats including MRs and SMRs.

Enumerating these context-sensitive repeats has numerous applications in bioinformatics~\cite{1997Gusfield_AlgorOnStrinTreesAnd,2008NicolasRSPCDTVM_ModelLocalRepeatOnGenom,2009BecherDH_EfficComputOfAllPerfec} and natural language processing~\cite{2001LinY_ExtracChinesFrequenStrinWithout,2004LinY_ProperAndFurthApplicOf,2009OkanoharaT_TextCategWithAllSubst_SDM,2011MasadaTSO_ClustDocumWithMaximSubst,2013Galle_BagOfRepeatRepresOf_SIGIR} especially valuable for languages with no clear word segmentation such as Chinese and Japanese.
Motivated by these applications, efficient enumeration algorithms have been extensively studied~\cite{2009OkanoharaT_TextCategWithAllSubst_SDM,2009BecherDH_EfficComputOfAllPerfec,2015OhlebuschB_AlphabIndepAlgorForFindin,2020BelazzouguiCKM_LinearTimeStrinIndexAnd,2021NishimotoT_REnumEnumerOfCharac_CPM} using a family of suffix data structures such as suffix trees~\cite{1973Weiner_LinearPatterMatchAlgor}, directed acyclic word graphs (DAWGs)~\cite{1985BlumerBHECS_SmallAutomRecogSubworOf}, compact DAWGs (CDAWGs)~\cite{1987BlumerBHME_ComplInverFilesForEffic,2005InenagaHSTAMP_LineConstOfCompacDirec}, suffix arrays~\cite{1993ManberM_SuffixArrayNewMethodFor,JDA2004AbouelhodaKO_ReplacSuffixTreesWithEnhan} and Burrows-Wheeler transform (BWT)~\cite{1994BurrowsW_BlockSortinLosslDataCompr}.

In this paper, we focus on \emph{r-enum}~\cite{2021NishimotoT_REnumEnumerOfCharac_CPM}, an algorithm to enumerate the MRs, \emph{minimal unique substrings (MUSs)}~\cite{2011IlieS_MinimUniqueSubstAndMaxim} and \emph{minimal absent words (MAWs)} in $O(r)$ words of compressed working space, where $r \le n$ is the number of runs in the BWT of $T$ of length $n$.
Given the \emph{run-length encoded BWT (RLBWT)} of $T$, r-enum runs in $O(n \log \log_{w} (n/r))$ time in addition to the time linear to the number of output strings, where $w = \Theta(\log n)$ is the word size.
We show that $O(n \log \log_{w} (n/r))$ time can be improved to $O(n)$ by resorting to the \emph{move data structure} recently proposed by Nishimoto and Tabei~\cite{2021NishimotoT_OptimTimeQueriesOnBwt_ICALP} for constant-time LF-mapping on RLBWTs.
We also extend r-enum to compute the NSMRs and SMRs, and the \emph{context diversity} for every RMR in the same time and space complexities.

Furthermore, we consider problems for net occurrences and net frequencies, originated from the studies of Chinese frequent strings~\cite{2001LinY_ExtracChinesFrequenStrinWithout,2004LinY_ProperAndFurthApplicOf} and have recently attracted attention from combinatorial and algorithmic points of view~\cite{2024GuoEWZ_ExploitNewProperOfStrin_CPM,2024OhlebuschBO_FasterComputOfChinesFrequen_SPIRE,2024GuoUWZ_OnlinComputOfStrinNet_SPIRE,2024Inenaga_FasterAndSimplOnlinComput,2025GuoK_NetOccurInFibonAnd_CPM,2025MienoI_SpaceEfficOnlinComputOf_CPM}.
An occurrence of a repeat is called a \emph{net occurrence} if the occurrence is not covered by another repeat,
and the \emph{net frequency} of a repeat is the number of its net occurrences.
We would like to point out the fact, which may have been overlooked for over two decades, that the strings with positive net frequencies are equivalent to the NSMRs~\cite{1997Gusfield_AlgorOnStrinTreesAnd} and the \emph{largest maximal repeats}~\cite{2008NicolasRSPCDTVM_ModelLocalRepeatOnGenom,2013Galle_BagOfRepeatRepresOf_SIGIR}, and thus appear to have broader applications.
Also, net occurrences are conceptually equivalent to \emph{maximum repeats} in~\cite{2011IlieS_MinimUniqueSubstAndMaxim}.

Given the RLBWT of $T$, we show how to compute the set $\nsmr$ of all NSMRs in $T$ together with their net frequency/occurrences in $O(n)$ time and $O(r)$ words of space, solving the task called \textsc{All-NF} in~\cite{2024GuoEWZ_ExploitNewProperOfStrin_CPM} on RLBWTs.
We also address the task called \textsc{Single-NF}~\cite{2024GuoEWZ_ExploitNewProperOfStrin_CPM}, which requires building a data structure to support NF-queries of computing the net frequency and/or net occurrences of any query pattern.
We show that an $O(r)$-space data structure can be built from the RLBWT to support NF-queries in optimal time.
The data structure is built in $O(r)$ space and in $O(n)$ time with high probability or deterministic $O(n + |\nsmr| \log \log \min(\sigma, |\nsmr|))$ time, where $\sigma \le r$ is the alphabet size of $T$.
To achieve this, we prove that the total number of net occurrences is less than $2r$.
With the duality between net occurrences and MUSs shown in~\cite{2011IlieS_MinimUniqueSubstAndMaxim,2025MienoI_SpaceEfficOnlinComputOf_CPM},
we immediately get a new upper bound $2r$ of the number of MUSs in $T$, which may be of independent interest.

\subsection{Related work}
Context-sensitive repeats are closely related to and efficiently computed from suffix data structures.
A folklore algorithm~\cite{1997Gusfield_AlgorOnStrinTreesAnd} uses suffix trees.
To reduce the (rather big) working space of suffix trees by a constant factor, algorithms based on enhanced suffix arrays~\cite{JDA2004AbouelhodaKO_ReplacSuffixTreesWithEnhan} have been proposed.
A compact- or compressed-space solution can be achieved by using LF-mapping and range distinct queries (see \cref{sec:prelim}) on BWTs or RLBWTs to simulate a traversal on suffix link trees.
\cref{table:compare_enum} summarizes representative algorithms.
\begin{table}[t]
  \centering
  \begin{tabular}{l|l|l|l}
  \begin{tabular}{@{}@{}l} reference:\\ used techniques\end{tabular} 
         & space (bits)
         & running time 
         & patterns \\ \hline\hline
  \cite{1997Gusfield_AlgorOnStrinTreesAnd}: suffix tree
         & $O(n \log n)$ 
         & $O(n \log \sigma)$
         & MR, NSMR, SMR \\ \hline
  \cite{2009OkanoharaT_TextCategWithAllSubst_SDM}: enhanced suffix array
         & $O(n \log n)$ 
         & $O(n)$
         & MR \\ \hline
  \cite{2015OhlebuschB_AlphabIndepAlgorForFindin}: enhanced suffix array
         & $O(n \log n)$ 
         & $O(n)$
         & \begin{tabular}{@{}@{}l} MR, NSMR, SMR,\\ context diversity \end{tabular} \\ \hline
  \cite{2020BelazzouguiCKM_LinearTimeStrinIndexAnd}: BWT and RD
         & $O(n \log \sigma)$
         & $O(n)$
         & MR, MUS, MAW \\ \hline
  \cite{2012BellerBO_SpaceEfficComputOfMaxim_SPIRE}: BWT and RD
         & $|\mathit{RD}| + O(n)$ 
         & $O(nd)$
         & MR, SMR \\ \hline
  \cite{2015BelazzouguiC_SpaceEfficDetecOfUnusual_SPIRE}: BWT and RD
         & $|\mathit{RD}| + O(\sigma^2 \log^2 n)$ 
         & $O(nd)$
         & MR, MUS, MAW\\ \hline
  \cite{2021NishimotoT_REnumEnumerOfCharac_CPM}: RLBWT and RD
         & $O(r \log n)$
         & $O(n \log \log_{w} (n/r))$
         & MR, MUS, MAW \\ \hline
  \begin{tabular}{@{}@{}l} this work:\\RLBWT and RD \end{tabular}
         & $O(r \log n)$
         & $O(n)$
         & \begin{tabular}{@{}@{}l} MR, MUS, MAW,\\ NSMR, SMR,\\ context diversity \end{tabular} \\  
  \end{tabular}

  \caption{Comparison between algorithms that find characteristic substring patterns in a string $T$ of length $n$ over an integer alphabet of size $\sigma = n^{O(1)}$, where $w = \Theta(\log n)$ is the word size.
  Each method takes $T$, its BWT or $O(r)$-size RLBWT as an input, and enumerate substring patterns represented in $O(1)$ space each.
  For MAWs, we additionally need the time linear to the number of output, which is bounded by $O(n\sigma)$.
  RD represents a data structure for range distinct queries, which can be implemented in $|\mathit{RD}|$ bits of space while supporting a query in $O(d)$ time per output with $(|\mathit{RD}|, d) = (n \log_2 \sigma + o(n \log \sigma), O(\log \sigma))$ using~\cref{lem:rdq} or $(|\mathit{RD}|, d) = (n \log_2 \sigma + o(n \log \sigma), O(\log \sigma))$~\cite{2015ClaudeNP_WavelMatrixEfficWavelTree}.}
  \label{table:compare_enum}
\end{table}

The study for net frequencies/occurrences was initiated recently by Guo et al.~\cite{2024GuoEWZ_ExploitNewProperOfStrin_CPM}.
Given a string $T$ of length $n$ over an integer alphabet of size $\sigma = n^{O(1)}$, we can perform \textsc{All-NF} in $O(n)$ time and words of space using an algorithm~\cite{2015OhlebuschB_AlphabIndepAlgorForFindin,2024OhlebuschBO_FasterComputOfChinesFrequen_SPIRE} based on enhanced suffix arrays.
The problem has also been studied in an online setting, where we maintain a data structure to perform \textsc{All-NF} and \textsc{Single-NF} while allowing monotonic updates to the left and/or right ends of $T$.
The current state-of-the-art online algorithm~\cite{2025MienoI_SpaceEfficOnlinComputOf_CPM} is based on online construction of implicit suffix trees and implicit CDAWGs.
In this paper, we achieve the first offline algorithm for \textsc{All-NF} and \textsc{Single-NF} that works in $O(r)$ words of compressed space.

\subsection{Organization of paper}
This paper is organized as follows.
In \cref{sec:prelim}, we introduce notation and tools we use in the following sections.
In \cref{sec:bound}, we prove that the total number of net occurrences is less than $2r$, which is crucial to bound the time and space complexities of our algorithms for net frequencies/occurrences.
In \cref{sec:renum}, we speed up r-enum and extend it to enumerate other context-sensitive repeats.
In \cref{sec:NFqueries}, we present an $O(r)$-space data structure for NF-queries based on RLBWTs.
In \cref{sec:conclusions}, we conclude the paper with some remarks.

\section{Preliminaries}\label{sec:prelim}
\subsection{Basic notation}
We assume a standard word-RAM model with word size $w = \Theta(\log n)$, where $n$ is the length of string $T$ defined later.
A space complexity is evaluated by the number of words unless otherwise noted.
Any integer treated in this paper is represented in $O(1)$ space, i.e., $O(w)$ bits.
A set of consecutive integers is called an \emph{interval}, which can be represented in $O(1)$ space by storing its beginning and ending integers.
For two integers $b$ and $e$, let $[b..e]$ denote the interval $\{ b, b+1, \dots, e \}$ beginning at $b$ and ending at $e$ if $b \le e$, and otherwise, the empty set.
We also use $[b..e) := [b..e-1]$ to exclude the right-end integer.

A \emph{string} (or an \emph{array}) $x$ over a set $S$ is a sequence $x[1]x[2] \cdots x[|x|]$, where $|x|$ denotes the length of $x$ and $x[i] \in S$ is the $i$-th element of $x$ for any $i \in [1..|x|]$.
The string of length $0$ is called the \emph{empty string} and denoted by $\emptystr$.
Let $S^{*}$ denote the set of strings over $S$, and let $S^{n} \subset S^{*}$ be the set of strings of length $n \ge 0$.
For integers $1 \le b \le e \le |x|$, the \emph{substring} of $x$ beginning at $b$ and ending at $e$ is denoted by $x[b..e] = x[b]x[b+1] \cdots x[e]$.
For convenience, let $x[b..e]$ with $b > e$ represent the empty string.
We also use $x[b..e) = x[b..e-1]$ to exclude the right-end character.
A \emph{prefix} (resp. \emph{suffix}) of $x$ is a substring that matches $x[1..e]$ (resp. $x[b..|x|]$) for some $e \in [0..|x|]$ (resp. $b \in [1..|x|+1]$).
We use the abbreviation $x[..e] = x[1..e]$ and $x[b..] = x[b..|x|]$.
For a non-negative integer $d$, let $x^{d}$ denote the string of length $|x|d$ such that $x^{d}[|x|(i-1)+1..i|x|] = x$ for any $i \in [1..d]$.

\subsection{Tools}\label{sec:tools}
A function that injectively maps a set $S$ of integers to $[1..O(|S|)]$ can be used as a \emph{compact dictionary} for $S$ to look up a value associated with an integer in $S$.
We use the following known result:
\begin{lemma}[\cite{2000Willard_ExaminComputGeometVanEmde,2008Ruzic_ConstEfficDictionInClose_ICALP}]\label{lem:dic}
  Given $S \subseteq [1..u]$ of size $m \le u \le 2^{O(w)}$, we can build an $O(m)$-size dictionary in $O(m)$ time with high probability or deterministic $O(m \log \log m)$ time so that lookup queries can be supported in $O(1)$ worst-case time.
\end{lemma}

A \emph{predecessor query} for an integer $i$ over the set $S$ of integers asks to compute $\max \{ j \in S \mid j \le i \}$, which can be efficiently supported by the data structure of~\cite[Appendix A]{2015BelazzouguiN_OptimLowerAndUpperBound}:
\begin{lemma}[\cite{2015BelazzouguiN_OptimLowerAndUpperBound}]\label{lem:pred}
  Given $S \subseteq [1..u]$ of size $m \le u \le 2^{O(w)}$, we can build an $O(m)$-size data structure in $O(m \log \log_{w}(u/m))$ time and $O(m)$ space to support predecessor queries in $O(\log \log_{w}(u/m))$ time for any $i \in [1..n]$.
\end{lemma}

A \emph{range distinct query} $\rdq(x, p, q)$ for a string $x$ and integers $1 \le p \le q \le |x|$ asks to enumerate $(c, p_c, q_c)$ such that $p_c$ and respectively $q_c$ are the smallest and largest positions in $[p..q]$ with $x[p_c] = x[q_c] = c$ for every distinct character $c$ that occurs in $x[p..q]$.
It is known that a range distinct query can be supported in output-optimal time after preprocessing $x$.
\begin{lemma}[\cite{2002Muthukrishnan_EfficAlgorForDocumRetriev_SODA,2020BelazzouguiCKM_LinearTimeStrinIndexAnd}]\label{lem:rdq}
  Given a string $x \in [1..\sigma]^{m}$, we can build an $O(m)$-size data structure in $O(m)$ time and space to support range distinct queries.
  With a null-initialized array of size $\sigma$, each query can be answered in $O(k)$ time, where $k$ is the number of outputs.
\end{lemma}
A minor remark is that the output of range distinct queries of~\cref{lem:rdq} may not be sorted in any order of characters.

\subsection{Context-sensitive repeats}\label{sec:csr}
Let $\Sigma_{\$} = [1..\sigma]$ be an ordered alphabet with the smallest character $\$$, and let $\Sigma = \Sigma_{\$} \setminus \{ \$ \}$.
Throughout this paper, $T[1..n] \in \Sigma_{\$}^{n}$ is a string of length $n \ge 2$ that contains all characters in $\Sigma_{\$}$.
\footnote{If $T$ is a string over a larger inefficient alphabet, we first replace every occurrence of a character with its rank in the character set actually used in $T$ to satisfy the assumption. Given the RLBWT of $T$ of size $r$, this task can be done in $O(n)$ time and $O(r)$ space using the sorting algorithm of~\cite[Lemma 13]{2020NishimotoT_FasterQueriesOnBwtRuns_X3}.}
We assume for convenience that $\$$ is used as a sentinel that exists at $T[n]$ and $T[0]$, and does not occur in $T[1..n-1]$.
A position $b \in [1..n]$ is called an \emph{occurrence} of a string $x \in \Sigma_{\$}^{*}$ in $T$ if $T[b..b+|x|) = x$.
Let $\occ(x)$ denote the set of occurrences of $x$ in $T$.
The \emph{left context} $\lcontext(x)$ (resp. \emph{right context} $\rcontext(x)$) of a substring $x \in \Sigma^{*}$ of $T$ is the set of characters immediately precedes (resp. follows) the occurrences of $x$, i.e., $\lcontext(x) = \{ T[b-1] \mid b \in \occ(x) \}$ (resp. $\rcontext(x) = \{ T[b+|x|] \mid b \in \occ(x) \}$).
The \emph{context diversity} of $x$ is the pair $\pqcd{|\lcontext(x)|}{|\rcontext(x)|}$ of integers in $[1..\sigma]$.
A substring $x$ of $T$ is called a \emph{repeat} if $|\occ(x)| \ge 2$, and called \emph{unique} if $|\occ(x)| = 1$.
A repeat $x \in \Sigma^{*}$ in $T$ is a \emph{left-maximal repeat (LMR)} (resp. \emph{right-maximal repeat (RMR)}) if $|\lcontext(x)| > 1$ (resp. $|\rcontext(x)| > 1$), and is a \emph{maximal repeat (MR)} if $x$ is both left- and right-maximal.
A repeat $x$ in $T$ is a \emph{supermaximal repeat (SMR)} if it is not a substring of another repeat.
An occurrence $b$ of a repeat $x$ is called a \emph{net occurrence} if there is no occurrence $b'$ of another repeat $x'$ such that $[b..b+|x|) \subset [b'..b'+|x'|)$.
Let $\nocc(x) \subseteq \occ(x)$ be the set of net occurrences of $x$ and let $\nfreq(x) := |\nocc(x)|$ be the number of net occurrences called the \emph{net frequency} of $x$.
A \emph{near-supermaximal repeat (NSMR)} is a repeat with a positive net frequency.
Let $\substr$, $\rep$, $\lmr$, $\rmr$, $\mr$, $\nsmr$ and $\smr$ be the sets of all substrings, repeats, LMRs, RMRs, MRs, NSMRs and SMRs in $T[1..n]$, respectively.

To describe algorithms uniformly, we assume that the empty string is a repeat that occurs at every position $b$ in $[1..n]$, i.e., $\occ(\emptystr) = [1..n]$.
Since $\lcontext(\emptystr) = \rcontext(\emptystr) = \Sigma_{\$}$, $\emptystr$ is always in $\lmr$, $\rmr$ and $\mr$.
We set the following exception for the definition of net occurrences of $\emptystr$:
An occurrence $b \in \occ(\emptystr)$ is a net occurrence of $\emptystr$ if and only if $T[b-1]$ and $T[b]$ are both unique.
This exception provides a smooth transition to another characterization of net occurrences based on the uniqueness of extended net occurrences~\cite{2024GuoEWZ_ExploitNewProperOfStrin_CPM,2025MienoI_SpaceEfficOnlinComputOf_CPM}.

Table~\ref{table:arrays} summarizes the concepts introduced in this subsection for $T = \idtt{abcbbcbcabc\$}$.

\begin{table}[t]
  \newcommand{\koko}[1]{\textcolor{red}{\underline{#1}}}
  \begin{tabular}{|l|l|l|l|l|c|}
  \hline
  $x \in \rep$ & occurrences & $\occ(x)$ & $\lcontext(x)$ & $\rcontext(x)$ & context diversity \\ \hline
  $\emptystr$ & \idtt{\koko{\,}a\koko{\,}b\koko{\,}c\koko{\,}b\koko{\,}b\koko{\,}c\koko{\,}b\koko{\,}c\koko{\,}a\koko{\,}b\koko{\,}c\koko{\,}\$\,}
    & $[1..12]$            & $\{\idtt{\$}, \idtt{a}, \idtt{b}, \idtt{c}\}$          & $\{\idtt{\$}, \idtt{a}, \idtt{b}, \idtt{c}\}$ & $(4, 4)$ \\
  \idtt{a}   & \idtt{\,\koko{a}\,b\,c\,b\,b\,c\,b\,c\,\koko{a}\,b\,c\,\$\,}
    & $\{1, 9\}$            & $\{\idtt{\$}, \idtt{c}\}$          & $\{\idtt{b}\}$ & $(2, 1)$ \\
  \idtt{b}   & \idtt{\,a\,\koko{b}\,c\,\koko{b}\,\koko{b}\,c\,\koko{b}\,c\,a\,\koko{b}\,c\,\$\,}
    & $\{2, 4, 5, 7, 10 \}$ & $\{\idtt{a}, \idtt{b}, \idtt{c}\}$ & $\{\idtt{b}, \idtt{c}\}$ & $(3, 2)$ \\
  \idtt{c}   & \idtt{\,a\,b\,\koko{c}\,b\,b\,\koko{c}\,b\,\koko{c}\,a\,b\,\koko{c}\,\$\,}
    & $\{3, 6, 8, 11 \}$    & $\{\idtt{b}\}$                     & $\{\idtt{\$}, \idtt{a}, \idtt{b}\}$ & $(1, 3)$ \\
  \idtt{ab}   & \idtt{\,\koko{a\,b}\,c\,b\,b\,c\,b\,c\,\koko{a\,b}\,c\,\$\,}
    & $\{1, 9\}$            & $\{\idtt{\$}, \idtt{c}\}$          & $\{\idtt{c}\}$ & $(2, 1)$ \\
  \idtt{bc}   & \idtt{\,a\,\koko{b\,c}\,b\,\koko{b\,c}\,\koko{b\,c}\,a\,\koko{b\,c}\,\$\,}
    & $\{2, 5, 7, 10\}$     & $\{\idtt{a}, \idtt{b}, \idtt{c}\}$ & $\{\idtt{\$}, \idtt{a}, \idtt{b}\}$ & $(3, 3)$ \\
  \idtt{cb}   & \idtt{\,a\,b\,\koko{c\,b}\,b\,\koko{c\,b}\,c\,a\,b\,c\,\$\,}
    & $\{3, 6\}$            & $\{\idtt{b}\}$                     & $\{\idtt{b}, \idtt{c}\}$ & $(1, 2)$ \\
  \idtt{abc}  & \idtt{\,\koko{a\,b\,c}\,b\,b\,c\,b\,c\,\koko{a\,b\,c}\,\$\,}
    & $\{1, 9\}$            & $\{\idtt{\$}, \idtt{c}\}$          & $\{\idtt{\$}, \idtt{b}\}$ & $(2, 2)$ \\
  \idtt{bcb}  & \idtt{\,a\,\koko{b\,c\,b}\,\koko{b\,c\,b}\,c\,a\,b\,c\,\$\,}
    & $\{2, 5\}$            & $\{\idtt{a}, \idtt{b}\}$           & $\{\idtt{b}, \idtt{c}\}$ & $(2, 2)$  \\  
  \hline
  \end{tabular}

  \begin{tabular}{|l|l|l|c|c|c|c|c|c|}
  \hline
  $x \in \rep$ & net occurrences & $\nocc(x)$ & $\nfreq(x)$ & LMR & RMR & MR & NSMR & SMR  \\ \hline
  $\emptystr$& \idtt{\,a\,b\,c\,b\,b\,c\,b\,c\,a\,b\,c\,\$\,}
    & $\emptyset$   & 0 & $\surd$ & $\surd$ & $\surd$ & & \\
  \idtt{a}   & \idtt{\,a\,b\,c\,b\,b\,c\,b\,c\,a\,b\,c\,\$\,}
    & $\emptyset$   & 0 & $\surd$ & & & & \\
  \idtt{b}   & \idtt{\,a\,b\,c\,b\,b\,c\,b\,c\,a\,b\,c\,\$\,}
    & $\emptyset$   & 0 & $\surd$ & $\surd$ & $\surd$ & & \\
  \idtt{c}   & \idtt{\,a\,b\,c\,b\,b\,c\,b\,c\,a\,b\,c\,\$\,}
    & $\emptyset$   & 0 & & $\surd$ & & & \\
  \idtt{ab}   & \idtt{\,a\,b\,c\,b\,b\,c\,b\,c\,a\,b\,c\,\$\,}
    & $\emptyset$   & 0 & $\surd$ &  & & & \\
  \idtt{bc}   & \idtt{\,a\,b\,c\,b\,b\,c\,\koko{b\,c}\,a\,b\,c\,\$\,}
    & $\{7\}$       & 1 & $\surd$ & $\surd$ & $\surd$ & $\surd$ & \\
  \idtt{cb}   & \idtt{\,a\,b\,c\,b\,b\,c\,b\,c\,a\,b\,c\,\$\,}
    & $\emptyset$   & 0 & & $\surd$ & & & \\
  \idtt{abc}  & \idtt{\,\koko{a\,b\,c}\,b\,b\,c\,b\,c\,\koko{a\,b\,c}\,\$\,}
    & $\{1, 9\}$    & 2 & $\surd$ & $\surd$ & $\surd$ & $\surd$ & $\surd$ \\
  \idtt{bcb}  & \idtt{\,a\,\koko{b\,c\,b}\,\koko{b\,c\,b}\,c\,a\,b\,c\,\$\,}
    & $\{2, 5\}$    & 2 & $\surd$ & $\surd$ & $\surd$ & $\surd$ & $\surd$ \\  
  \hline
  \end{tabular}

  \caption{An example of the concepts introduced in \cref{sec:csr} for $T = \idtt{abcbbcbcabc\$}$. For every repeat $x \in \rep$, we show a visualization of the occurrences of $x$, $\occ(x)$, $\lcontext(x)$, $\rcontext(x)$, and the context diversity of $x$ in the upper table.
  In the lower table, we show a visualization of the net occurrences of $x$, $\nocc(x)$, $\nfreq(x)$, and subsequently check if $x$ is in each of $\lmr$, $\rmr$, $\mr$, $\nsmr$ and $\smr$.}
  \label{table:csr}
\end{table}

\subsection{Suffix trees and suffix arrays}
The \emph{suffix trie} of $T$ is the tree defined on the set $\substr$ of nodes and edges from $x \in \substr$ to $xc \in \substr$ with $c \in \Sigma_{\$}$.
The \emph{suffix tree} of $T$ is the compacted suffix trie in which every non-branching internal node of the suffix trie is removed and considered to be an \emph{implicit node}.
Every internal node of the suffix tree is a right-maximal repeat $x$ and has $|\rcontext(x)|$ children.
A link from a node $x$ to a (possibly implicit) node $ax \in \substr$ for some $a \in \Sigma_{\$}$ is called a \emph{Weiner link}.
It is known that every node can be reached by following Weiner links from the root $\emptystr$, and the total number of Weiner links is at most $3n$, which can be seen from the duality between suffix trees and DAWGs~\cite{1985BlumerBHECS_SmallAutomRecogSubworOf}.

The \emph{suffix array} $\sufa[1..n]$ of $T$ is the integer array such that, for any $i \in [1..n]$, $T[\sufa[i]..]$ is the lexicographically $i$-th suffix among the non-empty suffixes of $T$.
For any string $x \in \substr$, the \emph{SA-interval} of $x$, denoted by $\sufi(x)$, is the maximal interval $[p..q]$ such that $x$ is a prefix of $T[\sufa[i]..]$ for any $i \in [p..q]$.
Then, it holds that $\occ(x) = \{ \sufa[i] \mid i \in \sufi(x)\}$.

\subsection{Burrows-Wheeler transform}
The \emph{Burrows-Wheeler transform (BWT)} $\bwt[1..n]$ of $T$ is the string such that $\bwt[i] = T[\sufa[i]-1]$ for any $i \in [1..n]$.
Note that $T[\sufa[i]-1] = T[0] = \$$ for $i$ with $\sufa[i] = 1$ due to the sentinel at $T[0]$.
By definition, it holds that $\lcontext(x) = \{ \bwt[i] \mid i \in \sufi(x) \}$ for any string $x$.
The LF-mapping $\LF$ is the permutation on $[1..n]$ such that $\LF(i)$ represents the lexicographic rank of the suffix $T[\sufa[i]-1..]$ if $\sufa[i] \neq 1$, and otherwise $\LF(i) = 1$.
The FL-mapping $\FL$ is the inverse mapping of $\LF$.
Given an integer $i$ in $\sufi(x)$ for a string $x$, we can compute $x$ by using FL-mapping $|x|$ times from $i$ because $x[k] = \bwt[\FL^{k}(i)]$ holds for any $k~(1 \le k \le |x|)$, where $\FL^{k}(i)$ is defined recursively $\FL^{k}(i) = \FL(\FL^{k-1}(i))$ with the base case $\FL^{1}(i) = \FL(i)$.

\cref{table:arrays} shows an example of $\sufa$, $\LF(\cdot)$, $\FL(\cdot)$ and $\bwt$ for $T = \idtt{abcbbcbcabc\$}$.
\cref{fig:st} also illustrates the suffix tree of $T = \idtt{abcbbcbcabc\$}$ and a relationship between Weiner links and BWTs.

\begin{table}[t]
  \centering
  \begin{tabular}{|c||c|c|c|c|l|}
  \hline
  \multicolumn{1}{|c||}{$i$} & \multicolumn{1}{c|}{$\sufa[i]$} & \multicolumn{1}{c|}{$\LF(i)$} & \multicolumn{1}{c|}{$\FL(i)$} & \multicolumn{1}{c|}{$\bwt[i]$} & \multicolumn{1}{c|}{$T[\sufa[i]..]$} \\ \hline
  1   & 12  & 9   & 3  & \idtt{c}   & \idtt{\$} \\
  2   & 9   & 10  & 5  & \idtt{c}   & \idtt{abc\$} \\
  3   & 1   & 1   & 7  & \idtt{\$}  & \idtt{abcbbcbcabc\$} \\
  4   & 4   & 11  & 8  & \idtt{c}   & \idtt{bbcbcabc\$} \\
  5   & 10  & 2   & 9  & \idtt{a}   & \idtt{bc\$} \\
  6   & 7   & 12  & 10 & \idtt{c}   & \idtt{bcabc\$} \\
  7   & 2   & 3   & 11 & \idtt{a}   & \idtt{bcbbcbcabc\$} \\
  8   & 5   & 4   & 12 & \idtt{b}   & \idtt{bcbcabc\$} \\  
  9   & 11  & 5   & 1  & \idtt{b}   & \idtt{c\$} \\
  10  & 8   & 6   & 2  & \idtt{b}   & \idtt{cabc\$} \\
  11  & 3   & 7   & 4  & \idtt{b}   & \idtt{cbbcbcabc\$} \\
  12  & 6   & 8   & 6  & \idtt{b}   & \idtt{cbcabc\$} \\
  \hline
  \end{tabular}
  \caption{An example of $\sufa[i]$, $\LF(i)$, $\FL(i)$ and $\bwt[i]$ for $T = \idtt{abcbbcbcabc\$}$.}
  \label{table:arrays}
\end{table}

\begin{figure}[t]
  \centering
  \begin{minipage}[b]{0.63\linewidth}
    \centering
    \includegraphics[scale=0.34]{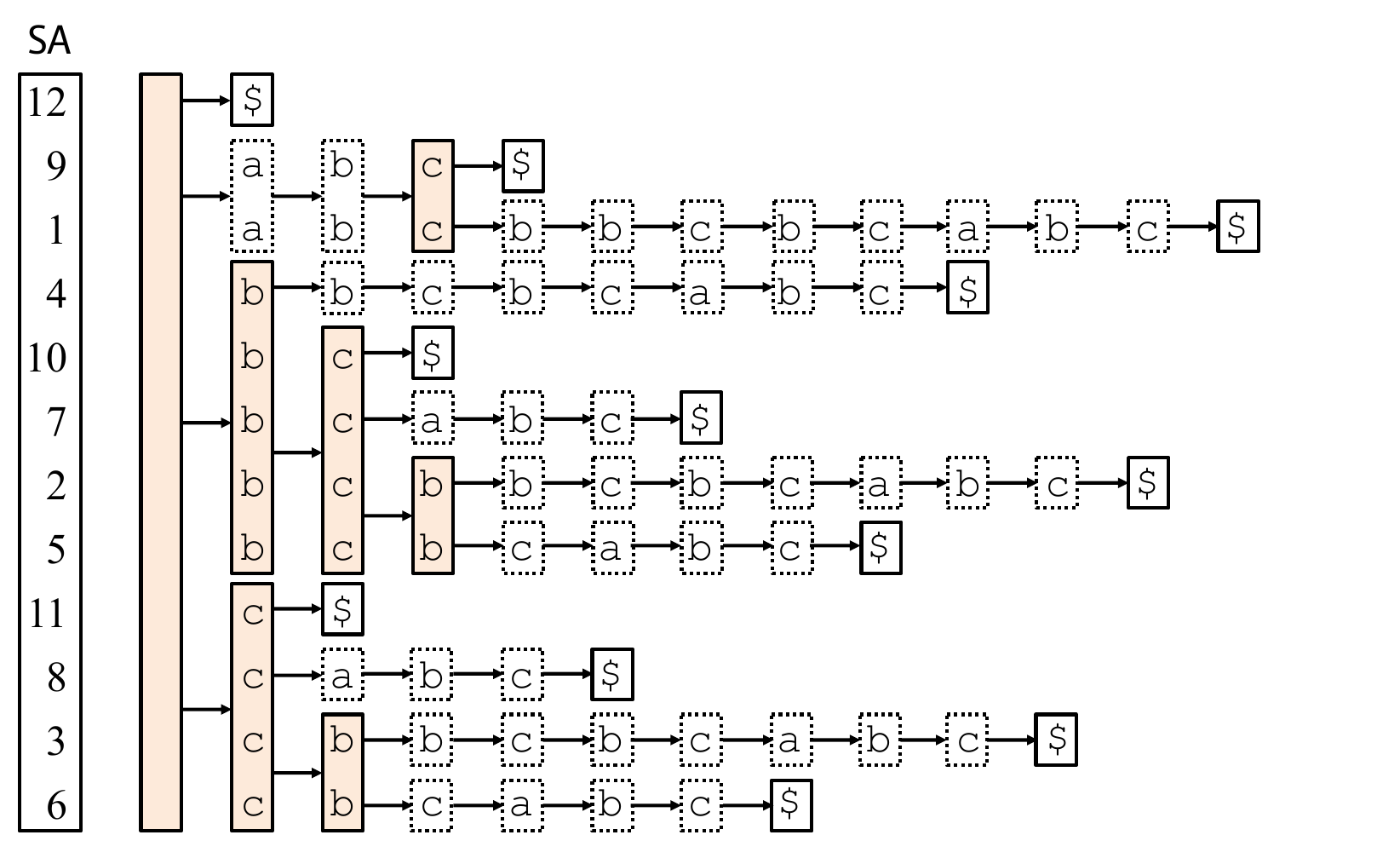}
  \end{minipage}
  \begin{minipage}[b]{0.35\linewidth}
    \centering
    \includegraphics[scale=0.34]{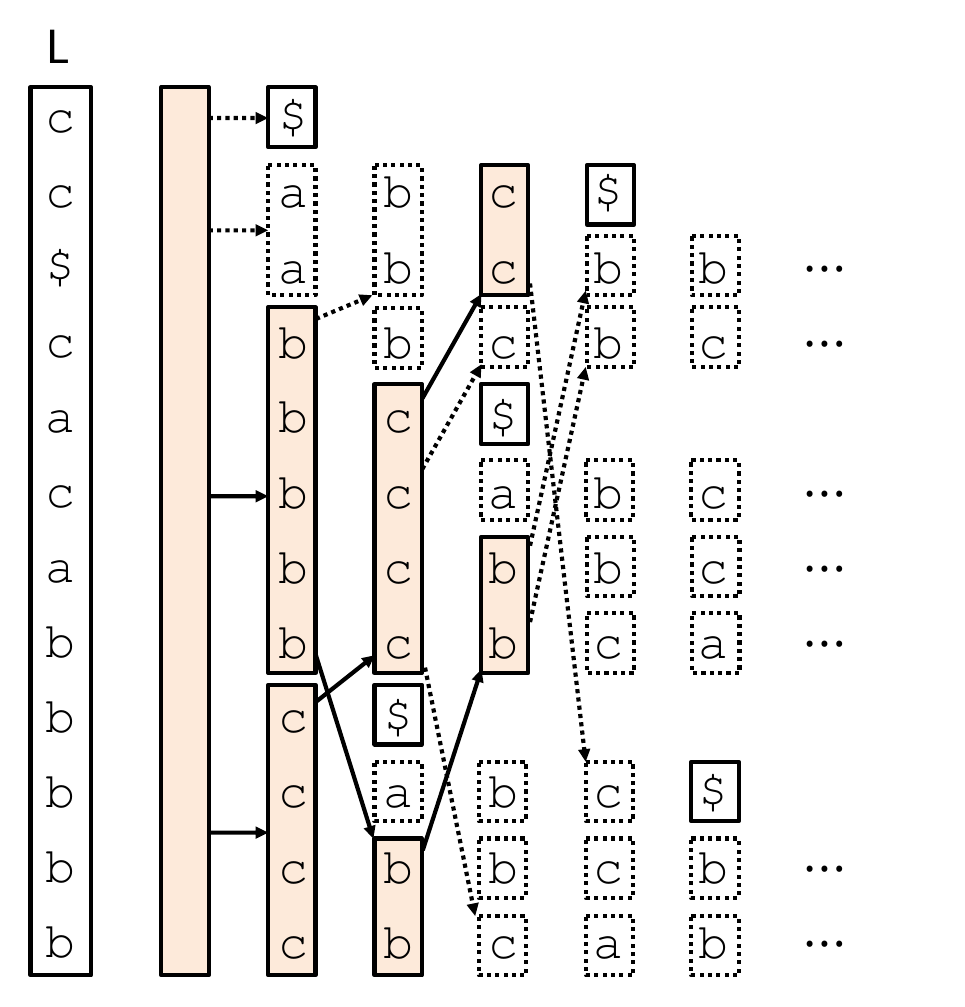}
  \end{minipage}
  \caption{
    The left figure shows the suffix tree of $T = \idtt{abcbbcbcabc\$}$ illustrated over sorted suffixes, on which each node $x$ can be represented by $\sufi(x)$ and $|x|$.
    A solid box is a node (highlighted for internal nodes), and a dotted box is an implicit node.
    The right figure shows the Weiner links outgoing from all internal nodes (the Weiner links from leaves are omitted). 
    The Weiner links that points to internal nodes are depicted with solid arrows, and the other ones with dotted arrows.
    Observe that there is a Weiner link from a node $x$ to $ax$ for any character $a \in \lcontext(x) = \{ \bwt[i] \mid i \in \sufi(x) \}$, where the case with $a = \$$ is excluded unless $x = \emptystr$.
  }
  \label{fig:st}
\end{figure}

The \emph{run-length encoded BWT (RLBWT)} of $T$ represents $\bwt$ in $O(r)$ space by string $\runc[1..r] \in \Sigma_{\$}^{r}$ and array $\rund[1..r] \in [1..n]^{r}$ such that $\bwt = \runc[1]^{\rund[1]} \runc[2]^{\rund[2]} \dots \runc[r]^{\rund[r]}$ and $\runc[i] \neq \runc[j]$ for any $1 \le i < j \le r$.
Each component $\runc[i]^{\rund[i]}$ is called a \emph{run} of $\bwt$.
While $r \le n$ is obvious, inequality $\sigma \le r$ holds under the assumption that $T$ contains all characters in $\Sigma_{\$}$.

\begin{example}
  For $T = \idtt{abcbbcbcabc\$}$, its BWT $\bwt = \idtt{cc\$cacabbbbb}$ has $r = 7$ runs.
  Then, we have $\runc = \idtt{c\$cacab}$ and $\rund = [2, 1, 1, 1, 1, 1, 5]$.
\end{example}

A notable property of $\LF$ is that $\LF(i) = k_1 + k_2$ holds, where $k_1$ is the number of characters smaller than $\bwt[i]$ in $\bwt$ and $k_2$ is the number of occurrences of $\bwt[i]$ in $\bwt[1..i]$.
In particular, if $\bwt[p..q]$ consists of a single character, then $\LF(i) = \LF(p) + i - p$ for any $i \in [p..q]$, and $\FL(i') = \FL(p') + i' - p'$ for any $i' \in [p'..q']$ with $p' = \LF(p)$ and $q' = \LF(q)$.
Hence, LF-mapping can be implemented in $O(r)$ space by constructing the predecessor data structure for the set of beginning positions of runs of $\bwt$ and remembering the $\LF(\cdot)$ values for those positions (FL-mapping can be implemented similarly).
Using the predecessor data structure of~\cref{lem:pred}, we get:
\begin{lemma}\label{lem:lfpred}
  Given the RLBWT of $T$, we can construct an $O(r)$-size data structure in $O(r \log \log_{w}(n/r))$ time and $O(r)$ space to compute $\LF(i)$ and $\FL(i)$ in $O(\log \log_{w}(n/r))$ time for any $i \in [1..n]$.
\end{lemma}

Nishimoto and Tabei~\cite{2021NishimotoT_OptimTimeQueriesOnBwt_ICALP} proposed a novel data structure, called the \emph{move data structure}, to implement LF-mapping by a simple linear search on a list of intervals without predecessor queries.
We consider a list $[p_1..p_2), [p_2..p_3), \dots, [p_{\hat{r}}..p_{\hat{r}+1})$ of $\hat{r} = O(r)$ intervals such that $p_1 = 1$, $p_{\hat{r}+1} = n+1$, and for any $k \in [1..\hat{r}]$, $\bwt[p_k..p_{k+1})$ consists of a single character and $[\LF(p_k)..\LF(p_{k+1}-1)]$ intersects with $O(1)$ intervals in the list.
We let each $[p_k..p_{k+1})$ have a pointer to the interval $[p_{k'}..p_{k'+1})$ that contains $\LF(p_k)$ and offset $o_k = \LF(p_{k}) - p_{k'}$.
Then, $\LF(i) = \LF(p_{k}) + i - p_k = p_{k'} + o_k + i - p_k$ holds for any $i \in [p_k..p_{k+1})$.
Given $i$ and interval $[p_k..p_{k+1})$ that contains $i$, we can compute $\LF(i)$ and find the interval that contains $\LF(i)$ by assessing $O(1)$ intervals from $[p_{k'}..p_{k'+1})$ to the right in the list.
Note that $\LF(i)$ computation on the move data structure works along with the pointer to the interval in the list that contains $i$, which is handled implicitly hereafter.
Under the assumption that $\sigma \le r$ in this paper, we can construct the data structure needed for LF-mapping (and FL-mapping) in $O(r \log r)$ time~\cite[Appendix D]{2020NishimotoT_FasterQueriesOnBwtRuns_X3}:
\begin{lemma}\label{lem:lfmove}
  Given the RLBWT of $T$, we can construct an $O(r)$-size data structure in $O(r \log r)$ time and $O(r)$ space to compute $\LF(i)$ and $\FL(i)$ in $O(1)$ time for any $i \in [1..n]$.
\end{lemma}

\section{Upper bound $2r$ of net occurrences and MUSs}\label{sec:bound}
In this section, we consider combinatorial properties of net occurrences on BWTs.

We start with the following lemma.
\begin{lemma}\label{lem:nocc_iff}
  For any integer $i \in [p..q] = \sufi(x)$ for a repeat $x$ of $T$,
  $\sufa[i]$ is a net occurrence of $x$ if and only if $\bwt[i]$ is unique in $\bwt[p..q]$ 
  and $T[\sufa[i]..\sufa[i]+|x|]$ is unique in $T$.
\end{lemma}
\begin{proof}
  If $a = \bwt[i]$ is not unique in $\bwt[p..q]$, then $\sufa[i]$ is not a net occurrence of $x$ because $ax = T[\sufa[i]-1..\sufa[i]+|x|)$ is a repeat that covers $[\sufa[i]..\sufa[i]+|x|)$.
  Similarly, if $T[\sufa[i]..\sufa[i]+|x|]$ is not unique in $T$, then $\sufa[i]$ is not a net occurrence of $x$ because $T[\sufa[i]..\sufa[i]+|x|]$ is a repeat that covers $[\sufa[i]..\sufa[i]+|x|)$.

  On the other hand, if $\bwt[i]$ is unique in $\bwt[p..q]$ and $T[\sufa[i]..\sufa[i]+|x|]$ is unique in $T$, 
  there is no repeat that covers $[\sufa[i]..\sufa[i]+|x|)$, and hence, $\sufa[i]$ is a net occurrence of $x$.
\end{proof}

\begin{example}
  Let us consider a repeat $\idtt{bc}$ with $\sufi(\idtt{bc}) = [5..8]$ in the example of \cref{table:arrays}.
  $\sufa[6] = 7$ is a net occurrence of $\idtt{bc}$ because $\bwt[6] = \idtt{c}$ is unique in $\bwt[5..8] = \idtt{acab}$, and $T[\sufa[6]..\sufa[6]+2] = T[7..9] = \idtt{bca}$ is unique in $T$.
  $\sufa[5] = 10$ and $\sufa[7] = 2$ are not net occurrences of $\idtt{bc}$ because $\bwt[5] = \bwt[7] = \idtt{a}$ is not a unique character in $\bwt[5..8]$.
  $\sufa[7] = 2$ and $\sufa[8] = 5$ are not net occurrences of $\idtt{bc}$ because $T[2..4] = T[5..7] = \idtt{bcb}$ is a repeat.
\end{example}

The following lemma was implicitly proved in~\cite[Theorem~24]{2024GuoEWZ_ExploitNewProperOfStrin_CPM}
to show the total length of all net occurrences is bounded by twice the sum of irreducible longest common prefixes.
\begin{lemma}\label{lem:nocc_rlbwt}
  If $\sufa[i]$ is a net occurrence of a repeat $x$ of $T$, then $i$ is at the beginning or ending position of a run in $\bwt$, 
  and $\sufa[i]$ cannot be a net occurrence of another repeat.
\end{lemma}
\begin{proof}
  Let $[p..q] = \sufi(x)$.
  If $i \in [p..q]$ is not at the run's boundaries of $\bwt$, then at least one of $\bwt[i] = \bwt[i+1]$ with $i+1 \in [p..q]$ or $\bwt[i-1] = \bwt[i]$ with $i-1 \in [p..q]$ holds, which means that $\sufa[i]$ is not a net occurrence of $x$ by \cref{lem:nocc_iff}.
  $\sufa[i]$ cannot be a net occurrence of two distinct repeats, since otherwise, the occurrence of the shorter repeat is covered by the longer one.
\end{proof}

\cref{lem:nocc_rlbwt} leads to the following theorem:
\begin{theorem}\label{theo:nocc_2r}
  The total number of net occurrences of $T$ is less than $2r$.
\end{theorem}
\begin{proof}
  Let $r_1$ and respectively $r_{>1}$ be the number of runs in $\bwt$ of length $1$ and longer than $1$, i.e., $r = r_1 + r_{>1}$.
  Since there are $r_1 + 2r_{>1}$ positions that correspond to beginning or ending positions of runs in $\bwt$, it is clear from \cref{lem:nocc_rlbwt} that the total number of net occurrences is at most $r_1 + 2r_{>1} \le 2r$.
  The number is shown to be less than $2r$ by looking closer at some special positions like $1$ and $n$ on $\bwt$:
  Since $\sufa[1]$ (resp. $\sufa[n]$) cannot be a net occurrence if the first (resp. last) run of $\bwt$ is longer than $1$, only one net occurrence is charged to the first (resp. last) run of $\bwt$.
\end{proof}

A non-empty interval $[b..e]$ is called a \emph{minimal unique substring (MUS)} if $T[b..e]$ is unique while $T[b+1..e]$ and $T[b..e-1]$ are repeats.
In~\cite{2011IlieS_MinimUniqueSubstAndMaxim,2025MienoI_SpaceEfficOnlinComputOf_CPM}, it has been proved that the net occurrences and MUSs are dual concepts, basically saying that
for any two consecutive MUSs $[b..e]$ and $[b'..e']$ with $b < b'$, $b+1$ is a net occurrence of $T[b+1..e'-1]$.
With this duality, we immediately get the following:
\begin{corollary}\label{coro:mus_2r}
  The number of MUSs of $T$ is less than $2r$.
\end{corollary}

\section{R-enum revisited}\label{sec:renum}

\subsection{Rough sketch of r-enum and speedup}
Conceptually, r-enum performs a breadth-first traversal on the internal nodes $\rmr$ of the suffix tree by following Weiner links using RLBWT-based data structures.
Each node $x \in \rmr$ is visited with its suffix interval $\sufi(x)$ and length $|x|$ together with the information about its right extensions $\rext(x)$, which is the list of the set $\{ (c, \sufi(xc)) \mid c \in \rcontext(x) \}$ sorted by the character $c$.
The triplet $\repr(x) = (\sufi(x), \rext(x), |x|)$ is called the \emph{rich representation}~\cite{2020BelazzouguiCKM_LinearTimeStrinIndexAnd,2021NishimotoT_REnumEnumerOfCharac_CPM} of $x$, which can be stored in $O(|\rcontext(x)|)$ space.

Given $\repr(x) = ([p..q], \{(c_i, [p_i..q_i])\}_{i = 1}^{k}, |x|)$, we can compute $\{ \repr(ax) \mid ax \in \substr, a \in \Sigma \}$ as follows:
A null-initialized array $A[1..\sigma]$ is used as a working space to build $\rext(ax)$ at $A[a]$ for $a \in \lcontext(x) \setminus \{ \$ \}$.
For every $i \in [1..k]$, we compute $\lcontext(xc_i)$ by a range distinct query on $\bwt[p_i..q_i]$, and for each $a \in \lcontext(xc_i) \setminus \{ \$ \}$ we compute $\sufi(axc_i)$ by LF-mapping and append $(c_i, \sufi(axc_i))$ to the list being built at $A[a]$.
After processing all $i \in [1..k]$, $A[a]$ becomes $\rext(ax)$.
Finally, we enumerate $\lcontext(x)$ by a range distinct query on $\bwt[p..q]$ to go through only the used entries of $A$ to compute $\repr(ax)$ and clear the entries.
For the next round of the breadth-first traversal on $\rmr$, we keep only the set $\{ \repr(ax) \mid ax \in \rmr \}$, discarding $\repr(ax)$ if $\rext(ax)$ contains only one element.

\begin{example}
  In the r-enum for $T = \idtt{abcbbcbcabc\$}$, we visit a repeat $\idtt{bc} \in \rmr$ with its rich representation $\repr(\idtt{bc}) = (\sufi(\idtt{bc}), \rext(\idtt{bc}), |\idtt{bc}|) = ([5..8], \{ (\$, [5..5]), (\idtt{a}, [6..6]), (\idtt{b}, [7..8]) \}, 2)$.
  With a null-initialized array $A[1..\sigma]$, we compute $\repr(\idtt{abc}) = ([2..3], \{ (\$, [2..2]), (\idtt{b}, [3..3]) \}, 3)$ at $A[\idtt{a}]$, $\repr(\idtt{bbc}) = ([4..4], \{ (\idtt{b}, [4..4]) \}, 3)$ at $A[\idtt{b}]$, and $\repr(\idtt{cbc}) = ([12..12], \{ (\idtt{a}, [12..12]) \}, 3)$ at $A[\idtt{c}]$.
\end{example}

R-enum executes the above procedure in $O(r)$ space using range distinct queries on $\runc[1..r]$ and RLBWT-based LF-mapping.
The breadth-first traversal on $\rmr$ can be made in $O(r)$ space because the space needed to store the rich representations for all strings in $\{ ax \in \substr \mid a \in \Sigma, x \in \rmr, |x| = t \}$ for any $t \in [0..n]$ was proved to be $O(r)$ in Section~3.3 of~\cite{2021NishimotoT_REnumEnumerOfCharac_CPM}.
The total number of outputs of range distinct queries is upper bounded by $O(n)$ because each output can be charged to a distinct Weiner link.
Similarly, the total number of LF-mapping needed is bounded by $O(n)$.

The original r-enum used \cref{lem:lfpred} to implement LF-mapping and runs in $O(n \log \log_{w} (n/r))$ time in total (excluding the time to output patterns for MAWs).
We show that this can be improved to $O(n)$ time.
\begin{theorem}\label{theo:r-enum}
  Given the RLBWT of $T$, r-enum runs in $O(r)$ working space and $O(n)$ time in addition to the time linear to the number of output strings.
\end{theorem}
\begin{proof}
  If $r = \Omega(n / \log n)$, the original r-enum with \cref{lem:lfpred} is already upper bounded by $O(n)$ because $O(n \log \log_{w} (n/r)) = O(n \log \log_{w} \log n)$ and $w = \Theta(\log n)$.
  For the case with $r = o(n / \log n)$, we use the move data structure of \cref{lem:lfmove} to implement LF-mapping, which can be constructed in $O(r \log r) = O(n)$ time.
  Since LF-mapping can be done in $O(1)$ time on the move data structure, the r-enum runs in $O(n)$ time.
\end{proof}

\subsection{Extension for NSMRs, SMRs, and context diversities}
In this subsection, we extend r-enum to compute the NSMRs and SMRs, and the context diversity for every RMR.
Since $\smr \subseteq \nsmr \subseteq \mr \subseteq \rmr$ holds, for enumerating SMRs (resp. NSMRs), it is enough to check every element in $\rmr$ if it is in $\smr$ (resp. $\nsmr$).

Recall that r-enum traverses $x \in \rmr$ while computing $\{ \repr(ax) \mid ax \in \substr \}$ from $\repr(x) = ([p..q], \{(c_i, [p_i..q_i])\}_{i = 1}^{k}, |x|)$.
First, we can compute the context diversity $\pqcd{|\lcontext(x)|}{|\rcontext(x)|}$ of $x \in \rmr$ because $|\rcontext(x)|$ is $k$ and $|\lcontext(x)|$ is the number of outputs of a range distinct query on $\bwt[p..q]$.
Since $x$ is an SMR if and only if $|\lcontext(x)| = |\rcontext(x)| = |\occ(x)| = q - p + 1$, the context diversity has enough information to choose SMRs from RMRs.
On the other hand, NSMRs cannot be determined from the context diversities.

To determine if $x \in \rmr$ is an NSMR or not, we check if there is a net occurrence of $x$ in $\{ \sufa[j] \mid j \in [p..q] \}$ using the array $A[1..\sigma]$ that stores $\rext(ax)$ at $A[a]$ for any $a \in \lcontext(x) \setminus \{ \$ \}$.
Thanks to \cref{lem:nocc_iff}, we only have to consider $\sufa[p_i]$ for intervals $[p_i..q_i]$ with $p_i = q_i$ as candidates of net occurrences.
For an interval $[p_i..q_i]$ with $p_i = q_i$, we check if $a = \bwt[p_i]$ is $\$$ or $\sufi(ax)$ computed in $A[a]$ is a singleton, and then, $\sufa[p_i]$ is a net occurrence if and only if the condition holds.
We can also compute $\nocc(x)$ if we store $\sufa[i]$ for all run's boundaries $i$, which can be precomputed using LF-mapping $O(n)$ times.

Since we can check the conditions for NSMRs and SMRs along with the traversal of $\rmr$, the time and space complexities of r-enum are retained.
\begin{theorem}\label{theo:csr}
  Given the RLBWT of $T$, we can compute the NSMRs and SMRs, and the context diversity for every RMR in $O(n)$ time and $O(r)$ working space.
  The net occurrences of output patterns can also be computed in the same time and space complexities.
\end{theorem}

\section{An $O(r)$-size data structure for NF-queries}\label{sec:NFqueries}
In this section, we tackle the problem called \textsc{Single-NF}~\cite{2024GuoEWZ_ExploitNewProperOfStrin_CPM} on RLBWTs:
We consider preprocessing the RLBWT of $T$ to construct a data structure of size $O(r)$ so that we can support \emph{NF-queries} of computing $\nfreq(P)$ for any query pattern $P$.

We define the \emph{reversed trie} for a set $S$ of strings as follows:
\begin{itemize}
  \item The nodes consist of the suffixes of strings in $S$.
  \item There is an edge from node $x$ to $ax$ with a character $a$.
\end{itemize}

\begin{theorem}
  Given the RLBWT of $T$, we can build an $O(r)$-size data structure in $O(r)$ space and in $O(n)$ time with high probability or deterministic $O(n + |\nsmr| \log \log \min(\sigma, |\nsmr|))$ time to support NF-queries for any query pattern $P$ in $O(|P|)$ time.
\end{theorem}
\begin{proof}
  We execute r-enum of \cref{theo:csr}, which actually performs a breadth-first traversal of the reversed trie for $\rmr$ in $O(n)$ time and $O(r)$ space.
  Our idea is to build the compacted reversed trie $\revtrie$ for $\nsmr$ during the traversal, which can be stored in $O(|\nsmr|) = O(r)$ space because $|\nsmr| \le 2r$ by \cref{lem:nocc_rlbwt}.
  For each node $x$ of $\revtrie$, which corresponds to a string in $\nsmr$ or a branching node, we store the string length $|x|$ and an integer $i_{x}$ in $\sufi(x)$ so that the edge label between $x$ and its parent $y$ can be retrieved by using FL-mapping $|x| - |y|$ times from $i_{x}$.
  For a node $x \in \nsmr$, we also store the net frequency of $x$.
  See \cref{fig:singlenf} for an illustration of $\revtrie$.
  
  We show that we can build $\revtrie$ in $O(n)$ time and $O(r)$ working space during the breadth-first traversal of r-enum.
  For each string depth $t$, we maintain the compacted reversed trie $\revtrie_{t}$ for the set $\rmr_t \cup \nsmr_{\le t}$, where $\rmr_t$ is the set of RMRs of length $t$ and $\nsmr_{\le t}$ is the set of NSMRs of length $\le t$.
  Since $|\rmr_t| = O(r)$ and $|\nsmr_{\le t}| = O(r)$, the size of $\revtrie_{t}$ is $O(r)$.
  In the next round of the breadth-first traversal, r-enum extends every string in $\rmr_t$ by one character to the left to get $\rmr_{t+1}$.
  We update $\revtrie_{t}$ to $\revtrie_{t+1}$ by processing $x \in \rmr_t$ as follows, where $A_x = \{ ax \mid ax \in \rmr_{t+1} \}$:
  \begin{itemize}
    \item If $|A_x| \ge 2$, we add new edges from $x$ to the strings in $A_x$.
    \item If $|A_x| = 1$, we add a new edge from $x$ to the string in $A_x$, and remove the non-branching node $x$ if $x$ is not an NSMR.
    \item If $|A_x| = 0$ and $x \notin \nsmr$, we remove the node $x$ and the edge to its parent $y$. We also remove $y$, if it becomes a non-branching node not in $\nsmr$.
  \end{itemize}
  In total, these tasks can be done in $O(n)$ time and $O(r)$ space along with r-enum of \cref{theo:csr}.
  For each branching node of $\revtrie$ with $k \le \min(\sigma, |\nsmr|)$ outgoing edges, we build a dictionary of \cref{lem:dic} in $O(k)$ space and $O(k)$ time w.h.p. or deterministic $O(k \log \log \min(\sigma, k))$ time so that we can decide which child to proceed by the first characters of edges in $O(1)$ time.
  The construction time of the dictionaries for all branching nodes is bounded by $O(|\nsmr|)$ time w.h.p. or deterministic $O(|\nsmr| \log \log \min(\sigma, |\nsmr|))$ time.

  Given a query pattern $P$, we read the characters of $P$ from right to left and traverse $\revtrie$ from the root in $O(|P|)$ time.
  If we can reach a node $P \in \nsmr$, we output $\nfreq(P)$ stored in the node.
\end{proof}
We can also support queries of computing the net occurrences of a query pattern $P$ in $O(|P|+|\nocc(P)|)$ time by storing $\nocc(x)$ for each node $x \in \nsmr$ in $\revtrie$, which can be stored in $O(r)$ space in total due to \cref{lem:nocc_rlbwt}.

\begin{figure}[t]
  \centering
  \includegraphics[scale=0.34]{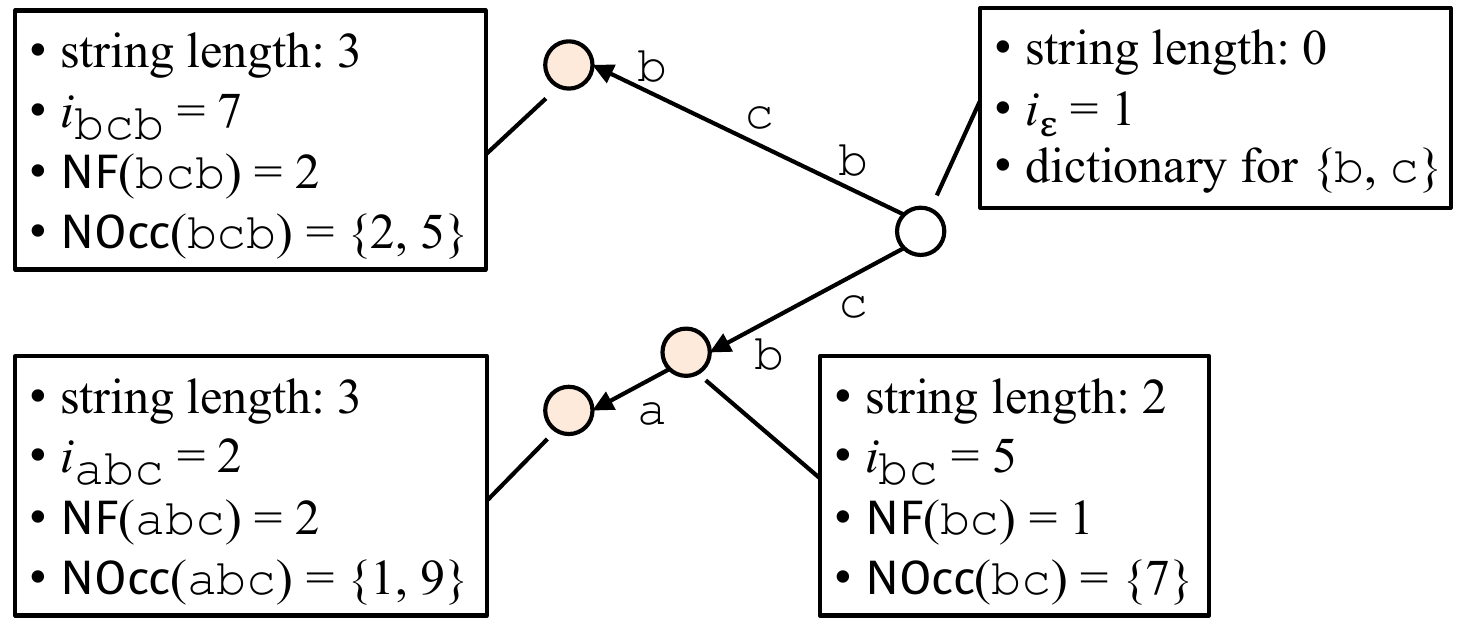}
  \caption{
    An illustration of the compacted reversed trie $\revtrie$ for $\nsmr = \{ \idtt{bc}, \idtt{abc}, \idtt{bcb} \}$ in our running example $T = \idtt{abcbbcbcabc\$}$.
    A node corresponding to an NSMR is highlighted.
    Storing $\nocc(\cdot)$ is optional.
    Note that edge labels are not stored explicitly.
    For example, the string $\idtt{bcb}$ on the edge from $\emptystr$ to $\idtt{bcb}$ is retrieved using FL-mapping $i_{\idtt{bcb}} - i_{\emptystr} = 3$ times from $i_{\idtt{bcb}} = 7$ when necessary.
  }
  \label{fig:singlenf}
\end{figure}

\section{Conclusions}\label{sec:conclusions}
We improved the time complexity of r-enum, and extended its enumeration power for other context-sensitive repeats in the literature.
We also proposed the first data structure of size $O(r)$ to support NF-queries in optimal time, which can be regarded as an $O(r)$-space representation of $\nsmr$ with the help of RLBWTs.
As a combinatorial property, we formally proved that the total number of net occurrences, as well as MUSs, is upper bounded by $O(r)$.

We note that our enumeration algorithm runs independently of the alphabet size $\sigma$ thanks to the output-optimal range distinct queries of~\cref{lem:rdq}.
The space usage depending on $\sigma$ is hidden behind $O(r)$.
We removed the dependence on $\sigma$ in the NF-query time by using a dictionary with worst-case constant-time lookup.

Finally, we remark that our work leads to the first algorithm to compute the sorted list of net occurrences and MUSs in $O(n)$ time and $O(r)$ space:
We enumerate all net occurrences by \cref{theo:r-enum}, and sort them in $O(n)$ time and $O(r)$ space by \cite[Lemma 13]{2020NishimotoT_FasterQueriesOnBwtRuns_X3}.

\bibliography{refs}

@InProceedings{1973Weiner_LinearPatterMatchAlgor,
	Title = {Linear pattern-matching algorithms},
	Author = {Peter Weiner},
	Booktitle = {Proc. 14th IEEE Ann. Symp. on Switching and Automata Theory},
	Year = {1973},
	Pages = {1-11},
}

@Article{1985BlumerBHECS_SmallAutomRecogSubworOf,
	Title = {The Smallest Automaton Recognizing the Subwords of a Text},
	Author = {Anselm Blumer and Janet Blumer and David Haussler and Andrzej Ehrenfeucht and M. T. Chen and Joel Seiferas},
	Journal = {Theoretical Computer Science},
	Year = {1985},
	Pages = {31--55},
	Volume = {40},
}

@Article{1987BlumerBHME_ComplInverFilesForEffic,
	author = {Anselm Blumer and J. Blumer and David Haussler and Ross M. McConnell and Andrzej Ehrenfeucht},
	title = {Complete inverted files for efficient text retrieval and analysis},
	journal = {J. {ACM}},
	volume = {34},
	number = {3},
	pages = {578--595},
	year = {1987},
}

@Article{1993ManberM_SuffixArrayNewMethodFor,
	Title = {Suffix Arrays: {A} New Method for On-Line String Searches},
	Author = {Udi Manber and Eugene W. Myers},
	Journal = {{SIAM} J. Comput.},
	Year = {1993},
	Number = {5},
	Pages = {935--948},
	Volume = {22},
}

@TechReport{1994BurrowsW_BlockSortinLosslDataCompr,
	author = {Michael Burrows and David J Wheeler},
	title = {A block-sorting lossless data compression algorithm},
	institution = {HP Labs},
	year = {1994},
}

@Article{2005InenagaHSTAMP_LineConstOfCompacDirec,
	Title = {On-line construction of compact directed acyclic word graphs},
	Author = {Shunsuke Inenaga and Hiromasa Hoshino and Ayumi Shinohara and Masayuki Takeda and Setsuo Arikawa and Giancarlo Mauri and Giulio Pavesi},
	Journal = {Discrete Applied Mathematics},
	Year = {2005},
	Number = {2},
	Pages = {156--179},
	Volume = {146},
}

@Book{1997Gusfield_AlgorOnStrinTreesAnd,
	Title = {Algorithms on Strings, Trees, and Sequences - Computer Science and Computational Biology},
	Author = {Dan Gusfield},
	Publisher = {Cambridge University Press},
	Year = {1997},
}

@Article{2009BecherDH_EfficComputOfAllPerfec,
	author = {Ver{\'{o}}nica Becher and Alejandro Deymonnaz and Pablo Ariel Heiber},
	title = {Efficient computation of all perfect repeats in genomic sequences of up to half a gigabyte, with a case study on the human genome},
	journal = {Bioinform.},
	volume = {25},
	number = {14},
	pages = {1746--1753},
	year = {2009},
}

@InProceedings{2009OkanoharaT_TextCategWithAllSubst_SDM,
	author = {Daisuke Okanohara and Jun'ichi Tsujii},
	title = {Text Categorization with All Substring Features},
	booktitle = {Proc. {SIAM} International Conference on Data Mining ({SDM}) 2009},
	pages = {838--846},
	year = {2009},
}

@InProceedings{2011MasadaTSO_ClustDocumWithMaximSubst,
	author = {Tomonari Masada and Atsuhiro Takasu and Yuichiro Shibata and Kiyoshi Oguri},
	title = {Clustering Documents with Maximal Substrings},
	booktitle = {Proc. 13th International Conference on Enterprise Information Systems {ICEIS} 2011},
	pages = {19--34},
	year = {2011},
}

@TechReport{2008NicolasRSPCDTVM_ModelLocalRepeatOnGenom,
	TITLE = {{Modeling local repeats on genomic sequences}},
	AUTHOR = {Nicolas, Jacques and Rousseau, Christine and Siegel, Anne and Peterlongo, Pierre and Coste, Fran{\c c}ois and Durand, Patrick and Tempel, S{\'e}bastien and Valin, Anne-Sophie and Mah{\'e}, Fr{\'e}d{\'e}ric},
	URL = {https://inria.hal.science/inria-00353690},
	TYPE = {Research Report},
	NUMBER = {RR-6802},
	PAGES = {43},
	INSTITUTION = {{INRIA}},
	YEAR = {2008},
}

@InProceedings{2013Galle_BagOfRepeatRepresOf_SIGIR,
	author = {Matthias Gall{\'{e}}},
	editor = {Gareth J. F. Jones and
                  Paraic Sheridan and
                  Diane Kelly and
                  Maarten de Rijke and
                  Tetsuya Sakai},
	title = {The bag-of-repeats representation of documents},
	booktitle = {Proc. 36th International {ACM} {SIGIR} conference on research and development in Information Retrieval 2013},
	pages = {1053--1056},
	publisher = {{ACM}},
	year = {2013},
}

@InProceedings{2012BellerBO_SpaceEfficComputOfMaxim_SPIRE,
	author = {Timo Beller and Katharina Berger and Enno Ohlebusch},
	editor = {Liliana Calder{\'{o}}n{-}Benavides and
                  Cristina N. Gonz{\'{a}}lez{-}Caro and
                  Edgar Ch{\'{a}}vez and
                  Nivio Ziviani},
	title = {Space-Efficient Computation of Maximal and Supermaximal Repeats in Genome Sequences},
	booktitle = {Proc. 19th International Symposium on String Processing and Information Retrieval ({SPIRE}) 2012},
	series = {Lecture Notes in Computer Science},
	volume = {7608},
	pages = {99--110},
	publisher = {Springer},
	year = {2012},
}

@Article{2015OhlebuschB_AlphabIndepAlgorForFindin,
	author = {Enno Ohlebusch and Timo Beller},
	title = {Alphabet-independent algorithms for finding context-sensitive repeats in linear time},
	journal = {J. Discrete Algorithms},
	volume = {34},
	pages = {23--36},
	year = {2015},
}

@InProceedings{2014GalleT_ContexDiverRepeatAndTheir_LATA,
	author = {Matthias Gall{\'{e}} and Mat{\'{\i}}as Tealdi},
	editor = {Adrian{-}Horia Dediu and
                  Carlos Mart{\'{\i}}n{-}Vide and
                  Jos{\'{e}} Luis Sierra{-}Rodr{\'{\i}}guez and
                  Bianca Truthe},
	title = {On Context-Diverse Repeats and Their Incremental Computation},
	booktitle = {Proc. 8th International Conference on Language and Automata Theory and Applications ({LATA}) 2014},
	series = {Lecture Notes in Computer Science},
	volume = {8370},
	pages = {384--395},
	publisher = {Springer},
	year = {2014},
}

@Article{2018GalleT_EmphxRepeatNewTaxonOf,
	author = {Matthias Gall{\'{e}} and Mat{\'{\i}}as Tealdi},
	title = {\emph{xkcd}-repeats: {A} new taxonomy of repeats defined by their context diversity},
	journal = {J. Discrete Algorithms},
	volume = {48},
	pages = {1--16},
	year = {2018},
}

@Article{2020BelazzouguiCKM_LinearTimeStrinIndexAnd,
	author = {Djamal Belazzougui and Fabio Cunial and Juha K{\"{a}}rkk{\"{a}}inen and Veli M{\"{a}}kinen},
	title = {Linear-time String Indexing and Analysis in Small Space},
	journal = {{ACM} Trans. Algorithms},
	volume = {16},
	number = {2},
	pages = {17:1--17:54},
	year = {2020},
}

@InProceedings{2002Muthukrishnan_EfficAlgorForDocumRetriev_SODA,
	Title = {Efficient algorithms for document retrieval problems},
	Author = {S. Muthukrishnan},
	Booktitle = {Proc. 13th Annual {ACM-SIAM} Symposium on Discrete Algorithms ({SODA}) 2002},
	Year = {2002},
	Pages = {657--666},
}

@Article{2015BelazzouguiN_OptimLowerAndUpperBound,
	author = {Djamal Belazzougui and Gonzalo Navarro},
	title = {Optimal Lower and Upper Bounds for Representing Sequences},
	journal = {{ACM} Trans. Algorithms},
	volume = {11},
	number = {4},
	pages = {31:1--31:21},
	year = {2015},
}

@Article{2000Willard_ExaminComputGeometVanEmde,
	author = {Dan E. Willard},
	title = {Examining Computational Geometry, Van Emde Boas Trees, and Hashing from the Perspective of the Fusion Tree},
	journal = {{SIAM} J. Comput.},
	volume = {29},
	number = {3},
	pages = {1030--1049},
	year = {2000},
}

@InProceedings{2008Ruzic_ConstEfficDictionInClose_ICALP,
	Title = {Constructing Efficient Dictionaries in Close to Sorting Time},
	Author = {Milan Ruzic},
	Booktitle = {Proc. 35th International Colloquium on Automata, Languages and Programming ({ICALP}) 2008},
	Year = {2008},
	Pages = {84--95},
}

@Article{2001LinY_ExtracChinesFrequenStrinWithout,
	author = {Yih{-}Jeng Lin and Ming{-}Shing Yu},
	title = {Extracting Chinese Frequent Strings Without Dictionary From a Chinese corpus, its Applications},
	journal = {J. Inf. Sci. Eng.},
	volume = {17},
	number = {5},
	pages = {805--824},
	year = {2001},
}

@Article{2004LinY_ProperAndFurthApplicOf,
	author = {Yih{-}Jeng Lin and Ming{-}Shing Yu},
	title = {The Properties and Further Applications of Chinese Frequent Strings},
	journal = {Int. J. Comput. Linguistics Chin. Lang. Process.},
	volume = {9},
	number = {1},
	year = {2004},
}

@InProceedings{2024GuoEWZ_ExploitNewProperOfStrin_CPM,
	author = {Peaker Guo and Patrick Eades and Anthony Wirth and Justin Zobel},
	editor = {Shunsuke Inenaga and
                  Simon J. Puglisi},
	title = {Exploiting New Properties of String Net Frequency for Efficient Computation},
	booktitle = {Proc. 35th Annual Symposium on Combinatorial Pattern Matching ({CPM}) 2024},
	series = {LIPIcs},
	volume = {296},
	pages = {16:1--16:16},
	publisher = {Schloss Dagstuhl - Leibniz-Zentrum f{\"{u}}r Informatik},
	year = {2024},
}

@InProceedings{2024OhlebuschBO_FasterComputOfChinesFrequen_SPIRE,
	author = {Enno Ohlebusch and Thomas B{\"{u}}chler and Jannik Olbrich},
	editor = {Zsuzsanna Lipt{\'{a}}k and
                  Edleno Silva de Moura and
                  Karina Figueroa and
                  Ricardo Baeza{-}Yates},
	title = {Faster Computation of Chinese Frequent Strings and Their Net Frequencies},
	booktitle = {Proc. 31st International Symposium on String Processing and Information Retrieval ({SPIRE}) 2024},
	series = {Lecture Notes in Computer Science},
	volume = {14899},
	pages = {249--256},
	publisher = {Springer},
	year = {2024},
}

@InProceedings{2024GuoUWZ_OnlinComputOfStrinNet_SPIRE,
	author = {Peaker Guo and Seeun William Umboh and Anthony Wirth and Justin Zobel},
	editor = {Zsuzsanna Lipt{\'{a}}k and
                  Edleno Silva de Moura and
                  Karina Figueroa and
                  Ricardo Baeza{-}Yates},
	title = {Online Computation of String Net Frequency},
	booktitle = {Proc. 31st International Symposium on String Processing and Information Retrieval ({SPIRE}) 2024},
	series = {Lecture Notes in Computer Science},
	volume = {14899},
	pages = {159--173},
	publisher = {Springer},
	year = {2024},
}

@Misc{2024Inenaga_FasterAndSimplOnlinComput,
	note = {arXiv:2410.06837},
	author = {Shunsuke Inenaga},
	title = {Faster and Simpler Online Computation of String Net Frequency},
	journal = {CoRR},
	volume = {abs/2410.06837},
	year = {2024},
}

@InProceedings{2025GuoK_NetOccurInFibonAnd_CPM,
	author = {Peaker Guo and Kaisei Kishi},
	editor = {Paola Bonizzoni and
                  Veli M{\"{a}}kinen},
	title = {Net Occurrences in Fibonacci and Thue-Morse Words},
	booktitle = {Proc. 36th Annual Symposium on Combinatorial Pattern Matching ({CPM}) 2025},
	series = {LIPIcs},
	volume = {331},
	pages = {16:1--16:22},
	publisher = {Schloss Dagstuhl - Leibniz-Zentrum f{\"{u}}r Informatik},
	year = {2025},
}

@InProceedings{2025MienoI_SpaceEfficOnlinComputOf_CPM,
	author = {Takuya Mieno and Shunsuke Inenaga},
	editor = {Paola Bonizzoni and
                  Veli M{\"{a}}kinen},
	title = {Space-Efficient Online Computation of String Net Occurrences},
	booktitle = {Proc. 36th Annual Symposium on Combinatorial Pattern Matching ({CPM}) 2025},
	series = {LIPIcs},
	volume = {331},
	pages = {23:1--23:13},
	publisher = {Schloss Dagstuhl - Leibniz-Zentrum f{\"{u}}r Informatik},
	year = {2025},
}

@Article{2011IlieS_MinimUniqueSubstAndMaxim,
	author = {Lucian Ilie and William F. Smyth},
	title = {Minimum Unique Substrings and Maximum Repeats},
	journal = {Fundam. Informaticae},
	volume = {110},
	number = {1-4},
	pages = {183--195},
	year = {2011},
}

@InProceedings{2021NishimotoT_REnumEnumerOfCharac_CPM,
	author = {Takaaki Nishimoto and Yasuo Tabei},
	editor = {Pawel Gawrychowski and
                  Tatiana Starikovskaya},
	title = {R-enum: Enumeration of Characteristic Substrings in {BWT}-runs Bounded Space},
	booktitle = {Proc. 32nd Annual Symposium on Combinatorial Pattern Matching ({CPM}) 2021},
	series = {LIPIcs},
	volume = {191},
	pages = {21:1--21:21},
	publisher = {Schloss Dagstuhl - Leibniz-Zentrum f{\"{u}}r Informatik},
	year = {2021},
}

@InProceedings{2021NishimotoT_OptimTimeQueriesOnBwt_ICALP,
	author = {Takaaki Nishimoto and Yasuo Tabei},
	editor = {Nikhil Bansal and
               Emanuela Merelli and
               James Worrell},
	title = {Optimal-Time Queries on {BWT}-Runs Compressed Indexes},
	booktitle = {Proc. 48th International Colloquium on Automata, Languages and Programming ({ICALP}) 2021},
	series = {LIPIcs},
	volume = {198},
	pages = {101:1--101:15},
	publisher = {Schloss Dagstuhl - Leibniz-Zentrum f{\"{u}}r Informatik},
	year = {2021},
}

@Misc{2020NishimotoT_FasterQueriesOnBwtRuns_X3,
	note = {arXiv:2006.05104v3},
	author = {Takaaki Nishimoto and Yasuo Tabei},
	title = {Optimal-Time Queries on {BWT}-Runs Compressed Indexes},
	volume = {abs/2006.05104},
	year = {2021},
}

@InProceedings{2015BelazzouguiC_SpaceEfficDetecOfUnusual_SPIRE,
	author = {Djamal Belazzougui and Fabio Cunial},
	editor = {Costas S. Iliopoulos and
               Simon J. Puglisi and
               Emine Yilmaz},
	title = {Space-Efficient Detection of Unusual Words},
	booktitle = {Proc. 22nd International Symposium on String Processing and Information Retrieval ({SPIRE}) 2015},
	series = {Lecture Notes in Computer Science},
	volume = {9309},
	pages = {222--233},
	publisher = {Springer},
	year = {2015},
}

@Article{2015ClaudeNP_WavelMatrixEfficWavelTree,
	author = {Francisco Claude and Gonzalo Navarro and Alberto Ord{\'{o}}{\~{n}}ez Pereira},
	title = {The wavelet matrix: An efficient wavelet tree for large alphabets},
	journal = {Inf. Syst.},
	volume = {47},
	pages = {15--32},
	year = {2015},
}

@Article{JDA2004AbouelhodaKO_ReplacSuffixTreesWithEnhan,
	Title = {Replacing suffix trees with enhanced suffix arrays},
	Author = {Mohamed Ibrahim Abouelhoda and Stefan Kurtz and Enno Ohlebusch},
	Journal = {Journal of Discrete Algorithms},
	Year = {2004},
	Number = {1},
	Pages = {53--86},
	Volume = {2}
}

\end{document}